\documentclass[11pt]{article}

\usepackage{amsmath,amsthm}
\usepackage[active]{srcltx}
\usepackage{graphicx}
\usepackage{amsmath,amssymb}
\usepackage{xspace}
\usepackage{algorithm}
\usepackage{algorithmic}
\usepackage{fontenc}
\usepackage{color}
\usepackage{setspace}
\renewcommand{\deg}{d}

\newcommand{\Oh}[1]{\mathcal{O}\left(#1\right)}

\newtheorem{thm}{Theorem}
\newtheorem{lem}[thm]{Lemma}
\newtheorem{cor}[thm]{Corollary}
\numberwithin{thm}{section}

\newcommand{\hit}{\mbox{\sf H}}
\newcommand{\MIX}{\mathsf{MIX}}

\newcommand{\nat}{\mathbb{N}}

\newcommand{\Pro}[1]{\ensuremath{\operatorname{\mathbf{Pr}}\left[#1\right]}}
\newcommand{\Ex}[1]{\ensuremath{\operatorname{\mathbf{E}}\left[#1\right]}}

\newcommand{\Bin}{\mathsf{Bin}}

\newcommand{\poly}{\ensuremath{\operatorname{poly}}}

\begin{document}

\title{Threshold Load Balancing in Networks}

\author{Martin Hoefer\thanks{Supported by DFG through Cluster of Excellence ``MMCI'' at Saarland University and grant Ho 3831/3-1. Max-Planck-Institute for Informatics and Saarland University, Saarbr\"ucken, Germany. {\tt mhoefer@mpi-inf.mpg.de}}
\and
  Thomas Sauerwald\thanks{Supported by DFG through Cluster of Excellence ``MMCI'' at Saarland University. Max-Planck-Institute for Informatics and Saarland University, Saarbr\"ucken, Germany. {\tt sauerwal@mpi-inf.mpg.de}}}

\date{}
\maketitle


\begin{abstract}
  We study probabilistic protocols for concurrent threshold-based load
  balancing in networks. There are $n$ resources or machines
  represented by nodes in an undirected graph and $m \gg n$ users
  that try to find an acceptable resource by moving along the edges of
  the graph. Users accept a resource if the load is below a
  \emph{threshold}. Such thresholds have an intuitive meaning, e.g.,
  as deadlines in a machine scheduling scenario, and they allow the
  design of protocols under strong locality constraints. When
  migration is partly controlled by resources and partly by users, our
  protocols obtain rapid convergence to a balanced state, in which all
  users are satisfied. We show that convergence is achieved in a number of 
  rounds that is only logarithmic in $m$ and polynomial in structural 
  properties of the graph. Even when migration is fully controlled by users, 
  we obtain similar results for convergence to approximately balanced states. 
  If we slightly adjust the migration probabilities in our protocol, we can 
  also obtain fast convergence to balanced states.
\end{abstract}
\thispagestyle{empty}

\clearpage
\setcounter{page}{1}
\section{Introduction}

Load balancing is a fundamental requirement of many distributed systems. The locality of information and communication inherent in many applications like multicore computer systems or wireless networks often render centralized optimization impossible. Instead, these cases require distributed load balancing algorithms that respect locality constraints, but nonetheless rapidly achieve balanced conditions. A successful approach to this problem are load balancing protocols, in which tasks are concurrently migrated in a distributed fashion. A variety of such protocols have been studied in the past, but they usually rely on machines to make migration decisions~\cite{Friedrich09,Rabani98,Elsaesser06,Elsaesser10}. Being a
fundamental resource allocation problem the interpretation of ``load''
and ``machine'' can greatly vary (e.g., in wireless networks it can
mean ``interference'' and ``channel''~\cite{Petrova10}, respectively), and in many cases machine-controlled reallocation represent an unreasonable means of centralized control. Protocols that avoid this feature have been popular in the area of algorithmic game theory. Here tasks are controlled by (selfish) users that follow a protocol to migrate their task to a less populated machine,
see~\cite{VoeckingChapter07}. While having distributed control, these protocols usually require strong forms of global knowledge, e.g., the number of underloaded/overloaded machines~\cite{EvenDar05}, or load differences among machines in the system~\cite{Berenbrink07,Berenbrink12,Berenbrink11,FotakisATOM10}. In many applications of interest, however, such information is unavailable or very difficult to obtain.

An interesting approach towards this problem are threshold-based load balancing protocols initially studied in~\cite{Fischer08}, in which reallocation decisions are based on an acceptance threshold. In the simplest variant, there is a uniform threshold $T$ and each user is satisfied if the machine it is currently assigned to has a number of assigned users below $T$. Otherwise, the user is dissatisfied and decides to migrate to another machine chosen uniformly at random. The great advantage of threshold-based protocols is that they can be implemented using only the information about the currently allocated machine and \emph{without having to obtain non-local information} about other machines, the current load or migration pattern in the system, etc. Successful balancing obviously also depends on a suitable threshold $T$.
The initial motivation for such an approach comes from wireless networks, where successful operation depends on acceptance thresholds for interference, and the amount of interference rises with the number of users on a channel. Here threshold protocols are necessary as distributed spectrum sensing (i.e., measuring the conditions of different channels) is a challenging and non-trivial problem by itself. In addition, protocols need to be user-controlled, because "machines" are frequency bands and cannot assign users to leave. Other interpretations of threshold balancing models include, e.g., thresholds as deadlines for the completion time of user tasks in a machine scheduling context. In this case, both resource-controlled or user-controlled protocols are of interest.

While threshold load balancing protocols are attractive, their behavior is not well-understood in many standard load balancing scenarios. In particular, previous works~\cite{Fischer08,AckermannFHS11} have only addressed the case when every machine is available to every user throughout the whole balancing process (i.e., when a ``complete network'' exists among resources). In this paper, we advance the understanding of these protocols in scenarios with locality restrictions to user migration. In particular, we assume that there is an undirected graph $G=(V,E)$ and each vertex $v \in V$ is a \emph{machine} or \emph{resource}. Users can access machines only \emph{depending on their location}, i.e., a user on machine $v \in V$ can only move to neighboring machines in $G$.  Load balancing with such an underlying neighborhood structure is frequently studied, but almost exclusively using machine-controlled protocols.

\paragraph{Contribution.}
We study protocols for threshold load balancing with resource-controlled and user-controlled migration. We assume thresholds are \emph{feasible}, i.e., they allow a \emph{balanced state} in which all users are satisfied and consider the expected convergence time to such a state. Our threshold model is very general and allows to capture a variety of scenarios. For instance, if thresholds represent deadlines in a machine scheduling environment, we can even assume machines to have set-up times and different speeds using an appropriate threshold for each machine and each user. For this case other protocols in the literature require strong means of coordination~\cite{FotakisATOM10,Elsaesser06} or achieve only pseudopolynomial convergence time~\cite{Berenbrink11}. In contrast, our protocols achieve rapid convergence in a number of rounds only logarithmic in the number $m$ of users and a polynomial depending on the graph structure. Hence, even in this very decentralized setting, efficient load balancing is still possible. The strong locality of the thresholds represents a challenge for the analysis, as many tools developed for, e.g., diffusion-based algorithms cannot directly be applied. In contrast, we here use potential function arguments in combination with analysis of random walks to prove convergence properties of our protocols.

After a formal introduction of our model in Section~\ref{sec:model}, we first concentrate in Section~\ref{sec:resource} on the case when user migration is partly controlled by resources. In particular, for each resource with dissatisfied users, we allow the resource to pick the users that should move. Each of the picked users then moves to an adjacent resource that it chooses uniformly at random. For \emph{user-independent thresholds}, where for every resource $v$ all users have the same threshold $T_v$, the protocol converges in $\Oh{\hit(G) \cdot \log(m)}$ rounds, where $\hit(G)$ is the maximum hitting time between any pair of nodes in $G$. If thresholds are arbitrary but satisfy an \emph{above average} property, the same holds and additionally the number of rounds is roughly in the order of $\Oh{\MIX(G)\cdot\log(m) + \hit(G) \cdot \log(n)}$. This is can be a much better bound as the mixing time $\MIX(G)$ of a random walk can be significantly smaller than $\hit(G)$ for many graphs $G$. This bound is shown to be essentially tight, as there are graphs $G$, above average thresholds, and initial allocations for which the protocol needs $\Omega(\MIX(G) \cdot \log(m))$ rounds to reach a balanced state. However, if we somewhat modify the protocol and start with slightly decreased above average thresholds, which are restored to their original value only after some time polynomial in $n$, we can avoid this lower bound and obtain a number of rounds independent of $m$.

In Section~\ref{sec:user} we consider a protocol that is fully user-controlled for the case of user-independent thresholds. In this case, each dissatisfied user independently at random decides to migrate to an adjacent resource with a probability depending on the locally observed loads and its intrinsic thresholds. When our aim is to balance approximately, we can establish similar bounds of $\Oh{\hit(G) \cdot \log(m)}$ and $\Oh{\MIX(G)\cdot\log(m)}$ for the cases of arbitrary and above average user-independent thresholds, respectively. To reach a completely balanced state, we only have to spend an additional $O(\poly(n))$ factor by decreasing the migration probability in the protocol.

All our results concern expected running times of the protocols. It is possible to obtain bounds that hold with high probability by spending an additional factor of $\log n$ in every bound. Details of this rather straightforward adjustment are omitted here.

\paragraph{Related Work.}
In algorithmic game theory several protocols for user-controlled selfish load balancing games have been proposed, using which a set of selfish users can reach a Nash equilibrium in a distributed and concurrent fashion. However, with the exception of~\cite{Berenbrink11} the protocols were studied only for the complete network. Some recent approaches are based on learning algorithms, but they allow to obtain only approximate stability and only as a distribution over states, even if we allow arbitrary finite time~\cite{Blum10,KleinbergPT11}. Protocols based on best response dynamics can converge much more rapidly. There are two approaches that yield convergence time of essentially $O(\log \log m + \poly(n))$, but either the number of underloaded/overloaded resources must be known~\cite{EvenDar05}, or users must be able to inspect load differences among resources in the system~\cite{Berenbrink07,Berenbrink12}. The latter is also necessary in~\cite{Berenbrink11,Adolphs12}, where the protocol from~\cite{Berenbrink07} is extended to arbitrary networks and convergence times of $O(\log(m) \cdot \poly(n))$ are shown. Inspection of load differences in the system is also central to protocols proposed for congestion games~\cite{Ackermann09, FotakisATOM10}.

Our threshold protocols that avoid this problem were proposed and analyzed for the complete network in~\cite{AckermannFHS11}, in which convergence in $O(\log(m))$ rounds is shown for both the resource- and user-controlled cases and user-independent and above average thresholds. We remark that there is an interpretation of our scenario as selfish load balancing game by assuming that each user experiences a private cost of 1 whenever the load on their allocated resource exceeds the threshold and 0 otherwise. In this way, our protocols can be interpreted to converge to Nash equilibria (i.e., the balanced states) of the game. For the case of resource-controlled migration, we assume that user thresholds are common knowledge. It is an interesting open problem to derive protocols for users private thresholds.

Load balancing with resource-controlled protocols has also received much interest in the distributed computing literature in recent years. The most prominent approaches are diffusion~\cite{Muthu98,Rabani98} and dimension-exchange models~\cite{Rabani98,Elsaesser10}, and the vast majority of the literature concentrates on the case of $m = n$ users. For this case, a wide variety of different bounds for general graphs and special topologies are known. However, in these models even the number of users that migrate from one resource to a specific (adjacent) resource is steered by the two resources.

We note that a different load balancing protocol based on random walks has been analyzed in~\cite{Elsaesser06,Elsaesser10}. However, the results there only hold for user-independent thresholds and assume a resource-controlled migration. In addition, not only dissatisfied users perform random walks, but also underloaded resources launch random walks to accelerate the balancing process.

\section{Model}
\label{sec:model}

\paragraph{Definition and Potential.}
There are $n$ \emph{machines} or \emph{resources}, which are nodes in a graph $G = (V,E)$, and a set $[m]$ of $m$ users. Each user has a unit-size task. It allocates the task to a resource and possibly moves along edges of the graph to find a resource with acceptable load. In particular, user $i$ has a \emph{threshold} $T_v^i$ for each resource $v \in V$. A \emph{state} is an assignment $a= (a_1,\ldots,a_m) \in V^m$ of users to resources. We let $x_v = |\{i \mid a_i = v\}|$ denote the \emph{load} on resource $v$, and we call $x$ the \emph{profile} of state $a$. If $i$ is assigned to $v$ and $x_v \le T_v^i$, then $i$ is happy with its choice. Otherwise, it is dissatisfied and motivated to leave. We consider distributed load balancing protocols to steer migration of dissatisfied users. We call a set of thresholds \emph{user-independent} if $T_v^i = T_v$ for all $i \in [m]$ and $v \in V$. We call thresholds \emph{resource-independent} if $T_v^i = T^i$ for all $i \in [m]$ and $v \in V$. Finally, we define the \emph{average} as $\overline{T} = \lceil m/n \rceil$, and call a set of thresholds \emph{above average} if $T_v^i > \overline{T}$ for all $i \in [m]$ and $v \in V$. For thresholds that are not user-independent (i.e., resource-independent or arbitrary) we will throughout make the assumption that they are above average. We call a state
\emph{balanced} if $x_v \le T_v^i$ for all resources $v$ and all users $i$ assigned to $v$. A set of thresholds $T_v^i$ is called \emph{feasible} if it allows a balanced state.

Many of our proofs are based on a potential function argument. We
define a \emph{potential} $\Phi(x) = \sum_{v \in V} \Phi_v(x)$ as
follows. Consider the users assigned to resource $v$ ranked in
non-increasing order of $T_v^i$. Let $k \in \{1,\ldots,x_v\}$ be the
last position in the ranking at which there is a user $i$ such that $k
\le T_v^i$. If there is no such position, we let $k=0$. The
contribution to the potential is $\Phi_v(x) = x_v - k$. Observe that
if thresholds are user-independent, $\Phi_v(x) = \max\{x_v - T_v,0\}$.

\paragraph{Random Walks.}
For an undirected, connected graph $G$, let $\Delta$, $d$, and $\delta$ be the maximum, average, and minimum degree of $G$, respectively. For a node $v \in V$, $\deg(v)$ is the degree of node $v$. If a user is continuously dissatisfied, its movements will form a random walk. The transition matrix of the random walk is the $n \times n $-matrix $\mathbf{P}$ which is defined by $P_{u,v}:= \frac{1}{\deg(u)}$ for $\{u,v\} \in E$ and $P_{u,v}:=0$ otherwise. Hence, the random walk moves in each step to a randomly chosen neighbor. Let $\mathbf{P}^t$ be the $t$-th power of $\mathbf{P}$. Then $P_{u,v}^t$ is the probability that a random walk starting from $u$ is located at node $v$ at step $t$. We denote by $\lambda_1 \geq \lambda_2 \geq \ldots \geq \lambda_n$ the $n$ eigenvalues of $\mathbf{P}$. We now define
\[ \mu:= 1 - \max_{2 \leq i \leq n} \left\{ |\lambda_i| \colon |\lambda_i| < 1
\right\}. \] 
(this definition differs slightly from the one of the spectral gap which is $1 - \max_{2 \leq i \leq n} |\lambda_i|$.) We further denote the stationary distribution of the random walk by the vector $\pi$ with $\pi_i = \deg(i)/(2|E|)$, where $m$ is the number of edges in $G$. For connected graphs, this distribution is the unique vector that satisfies $\pi \cdot \mathbf{P} =\pi$. However, the distribution of the random walk does not converge towards $\pi$ on bipartite graphs (as opposed to non-bipartite graphs). Therefore, the next lemma has to distinguish between bipartite and non-bipartite graphs. For a proof see the Appendix.
\begin{lem}
\label{lem:mixinglemma}
Let $G$ be any graph, $u,v \in V$ be any two nodes and $t \geq 4
\log (n) / \mu $.
\begin{itemize}
 \item If $G$ is non-bipartite, then $P_{u,v}^t = \pi(v) \pm n^{-3}$.
 \item If $G$ is bipartite with partitions $V_1,V_2$, then
\begin{equation*}
P_{u,v}^t =
\begin{cases}
 \pi(v) \cdot (1+(-1)^{t+1}) \pm n^{-3} & \mbox{if $u \in
V_1, v \in V_2$ or $u \in V_2, v \in V_1$}, \\
 \pi(v) \cdot (1+(-1)^{t}) \pm n^{-3} & \mbox{if $u \in
V_1, v \in V_1$ or $u \in V_2, v
\in V_2$}.
\end{cases}
\end{equation*}
\end{itemize}
\end{lem}

Based on Lemma~\ref{lem:mixinglemma}, we define the \emph{mixing time} to be
$\MIX(G):=4 \log n/\mu$. It is a well-known fact that $1/\mu$ is always at most polynomial in $n$, for instance, using the conductance we have $\MIX(G) \leq 4 n^4 \log n$. Note that if the random walks are made lazy, i.e., every walk stays at the current node with a loop probability $1-\alpha \in (0,1)$, then the above lemma applies similarly -- the only difference is that $\mu$ may be decreased by a factor of at most $\alpha$. Also, the case of bipartite graphs becomes subsumed by the case of non-bipartite graphs, because lazy walks do not suffer from bipartite oscillation.

We denote the \emph{hitting time} by $\hit(u,v)$ which is the expected time for a random walk to reach $v$ when starting from $u$ ($\hit(u,u)=0$). We define the maximum hitting time as $\hit(G):=\max_{u,v \in V} \hit(u,v)$.  For further details about random walks and mixing and hitting times, see
e.g.~\cite{Levin09,Lovasz93}.

\section{Resource-Controlled Migration}
\label{sec:resource}
In this section we consider a protocol with migration being partly resource- and partly user-controlled. In each round, every resource $v$ decides which of its assigned users to evacuate. The evacuation choice of the resource is done in accordance with the definition of the potential. Users currently assigned to $v$ are ordered in non-increasing order of $T_v^i$. Let $k$ be the last position in the ranking at which there is a user $i$ with $k \le T_v^i$. All users ranked after $i$ are assigned to leave the resource. Each user that is assigned to leave picks a neighboring resource uniformly at random and moves to this resource. All movements are concurrent, and there is no coordination between resources. A round ends when all users have moved and each resource has updated its sorted list of currently allocated users. Note that this protocol tries to accommodate as many users as possible on the resource and assigns exactly $\Phi_v(x)$ many users to leave.

For our analysis, we split a single round into two phases -- a
\emph{removal phase}, where resources remove the users to be evacuated
and an \emph{arrival phase}, where users arrive on their new
resources. After the removal phase, all remaining users are
satisfied. For the analysis of the arrival phase, we assume users
arrive sequentially on their chosen resources. If the arrival of a
single user does not cause the resource to remove an additional user
in the next round, this essentially reduces the potential by
1. Otherwise, either it is evacuated again in the next round, because
he is ranked too low at its new resource, or it causes at most one
user to migrate from the new resource.
Hence, in one round of the resource-controlled protocol the \emph{potential $\Phi(x)$ does not increase.}

This insight allows us to view migrating users as random walks. We assume a token is given in the arrival phase from a migrating user to the user it causes to migrate in the next round. The number of tokens in the system for state $a$ with load profile $x$ is exactly $\Phi(x)$, and each token performs a random walk over $G$. If a user causes no other user to migrate in the next round, the token is removed and the random walk is stopped. This reformulation of user migration is used in the proof of our general bounds on the convergence time for thresholds that are user-independent or above average.
\begin{thm}
  \label{thm:centralHit}
  For feasible user-independent or above average thresholds, the protocol
  converges to a balanced state in an expected number of $\Oh{\hit(G) \cdot
  \ln(m)}$ rounds.
\end{thm}
The main idea of the proof captured by the following lemma is to show that every $\hit(G)$ rounds a constant fraction of random walks is stopped. 
\begin{lem}
  \label{lem:hitProgress}
  From every starting state with load profile $x$ and $\Phi(x) > 0$ we
  reach after an expected number of rounds $\Ex{R} = \Oh{\hit(G)}$ a
  state with load profile $x^R$ such that
  $
  \Phi(x^R) \le \frac{3}{4} \cdot \Phi(x). 
  $
\end{lem}

\begin{proof}[Proof of Lemma~\ref{lem:hitProgress}]
  Consider the starting profile $x$ and first assume we want to move
  all $\Phi(x)$ random walks such that the potential reduces to 0. For
  this purpose, we consider for user-independent thresholds an
  arbitrary balanced state with load profile $x'$ in which all users
  assigned for migration are placed on resources that can accommodate
  them. For above average thresholds, consider $x'$ where each
  resource has load at most $\overline{T}$. Note that in each case the
  considered state has $\Phi(x') = 0$. We call a resource $v$
  underloaded if $x_v < T_v$ for user-independent thresholds and
  define $h_v = \max\{0, T_v - x_v\}$ as a lower bound on the number
  of users that can still be allocated to the underloaded resource $v$
  without creating dissatisfied users. Similarly, for above average
  thresholds $v$ is underloaded when $x_v < \overline{T}$ and define
  $h_v = \max\{0, \overline{T} - x_v\}$. Intuitively, we can think of
  $h_v$ as the number of ``holes'' in a balanced profile. It is easy
  to see that for feasible user-independent and above-average
  thresholds $\Phi(x) \le \sum_{v \in V} h_v(x)$.

  Hence, in the starting state we match each random walk token to a
  hole, which it should reach. In particular, we create a complete
  bipartite graph of $\Phi(x)$ nodes in one partition and $\sum_{v \in
    V} h_v(x)$ nodes in the other partition and compute a maximum
  matching. In this way, each token gets a resource and a specific
  position on this resource that it should occupy. If users carrying
  the tokens were reassigned according to the matching to their
  positions, a balanced state would be reached. We consider the random
  walks of tokens starting from $x$ and bound the time needed such
  that at least half of the tokens reach their corresponding
  destination resources at least once.

  Let us first assume that all the tokens are doing independent random
  walks that never stop. The expected time until token $t$ reaches its
  destination $v_t$ for the first time is $\hit(G)=\max_{u,v} \hit(u,v)$. Note
that after $2 \hit(G)$ rounds, the
  probability that $t$ has not reached the resource is at most $1/2$
  by Markov inequality. Now define the Bernoulli variable $R_t$ to be 1 if
token $t$ has reached
  $v_t$ after $2\hit(G)$ rounds. Let us say we are \emph{ready} when
  at least half of the tokens have visited their destination at least
  once. In particular, we are ready after $2\hit(G)$ rounds if $\sum_t
  R_t \ge \Phi(x)/2$. Using a Chernoff bound, we see that the
  probability of this event is at least $1-e^{-\Phi(x)/16} \geq 1-c$, for
  some constant $c < 1$. Thus, if we have not successfully brought
  $\Phi(x)/2$ random walks to their destination at least once, we
  restart the process. As $k$ restarts happen only with probability at
  most $c^{-k}$, the expected number of restarts is constant. Hence,
  in expectation $\Oh{\hit(G)}$ rounds are needed to bring at least
  half of the tokens to their destinations at least once.

  Until now, we have assumed that tokens always keep moving. In our
  real process, however, random walks might be stopped early because
  the tokens get removed on their way. This happens when they reach a
  resource where the user in the arrival phase does not increase the
  potential. Whenever this happens, we account the potential decrease
  of 1 towards the removed token. In contrast, a token $t$ might also
  reach the desired resource $v_t$, but does not stop moving, because
  other tokens have reached $v_t$ earlier and filled all available
  holes. Then, however, for each such token $t$ there is one other
  token $t'$ that has taken the spot of $t$ and thereby got
  removed. In this case, we reaccount 1/2 of the potential decrease
  for $t'$ towards $t$. Thus, every token that reaches its destination
  accounts for potential progress of at least 1/2. Hence, after
  $O(\hit(G))$ time in expectation, the potential has decreased by at
  least a fraction of 1/4, and thus the period has ended. This proves
  the lemma.
\end{proof}

\begin{proof}[Proof of Theorem~\ref{thm:centralHit}]
  We consider the convergence time from an arbitrary initial state with 
  profile $x$ to a balanced state in \emph{periods}. Period $j$ is the set of 
  rounds $r$, in which $\lceil\Phi(x)\cdot (3/4)^{j-1}\rceil \ge \Phi(x^r) >
  \lceil \Phi(x) \cdot (3/4)^j \rceil$, for $j=1$ until $\lceil
  \Phi(x) \cdot (3/4)^j \rceil = 1$. The last period begins when
  $\Phi(x^r) = 1$ and ends when a balanced state is
  reached. Obviously, there are in total $\Oh{\ln \Phi(x)}$ periods,
  which is at most $\Oh{\ln m}$. Lemma~\ref{lem:hitProgress} shows
  that the expected length of each period is $\Oh{\hit(G)}$. This
  proves the theorem.
\end{proof}

The following theorem presents a possibly improved bound if all resources have above average thresholds. For these thresholds we define $\varepsilon_{\min} = \min_{i,v} (T_v^i/\overline{T}) - 1$ as the minimum relative surplus over $\overline{T}$. In particular, for all the thresholds we have $T_v^i \ge (1+\varepsilon_{\min})\cdot \overline{T}$.
\begin{thm}
  \label{thm:centralMix}
  For above average thresholds the protocol converges to a balanced
  state in an expected number of rounds of
  \[
  \Oh{\left(\frac{1}{\varepsilon_{\min}} \cdot
      \frac{d}{\delta} \cdot \MIX(G) \cdot
      \ln(m)\right) + (\hit(G) \cdot \ln(n))}\enspace.
  \]
\end{thm}
Our proof below is based on the following lemma that bounds the time for a
significant number of random walks to reach a roughly balanced allocation.
\begin{lem}
  \label{lem:ballsBins}
  Consider $k \geq 192 \cdot \frac{2 |E|}{\delta} \log n$ random walks
  that start at arbitrary positions in $V$. For each resource $v \in V$, let $X=X(v) $ be the number of visits of all these random walks to $v$ at step $\MIX(G)$ and $\MIX(G)+1$.
  Then with probability at least $1-n^{-3}$, it holds for each $v \in V$ that
 \[
  X(v) \; \geq \; \frac{1}{2} \cdot \pi(v) \cdot k \; \geq \; \frac{1}{2} \cdot
  \frac{\delta}{2|E|} \cdot k.
\]
\end{lem}

\begin{proof}
  We first prove the statement for non-bipartite graphs. By
  Lemma~\ref{lem:mixinglemma} we have for $t=\MIX(G)$,
  \begin{equation}
   P_{u,v}^t \geq \pi(v) - \frac{1}{n^3}\enspace. \label{eq:lower}
  \end{equation}
  Hence, the number of random walks on resource $v$ at round $t$ can be
  written as a sum of independent, binary random variables,
  $X=X(v):=\sum_{i=1}^{k} X_i$, where $\Pro{ X_i = 1} \geq \pi(v) -
  \frac{1}{n^3}$. Therefore, $X$ is stochastically larger than
  $Y:=\Bin(x,\pi(v) - \frac{1}{n^3})$. Therefore, we may apply
  Lemma~\ref{lem:chernoff} to conclude that
  \begin{align*}
    \Pro{        X \leq \frac{1}{2} \pi(v) \cdot k } &\leq \Pro{ Y \leq \Ex{Y} -
\frac{1}{2} \pi(v) \cdot k + \frac{1}{n^3} \cdot k }   \\
                 &\leq \exp \left( - \frac{        ( \frac{1}{2} \pi(v) \cdot k -
\frac{1}{n^3} \cdot k)^2        }{ 2 ( k \pi(v) - \frac{k}{n^3} + (
\frac{1}{2} \pi(v) \cdot k - \frac{1}{n^3} \cdot k)    /3)                }
\right) \\
                 &\leq \exp \left(-        \frac{        \frac{1}{16} \pi(v)^2 k^2        }{3 \pi(v) k
}                        \right) \\
                 &\leq \exp \left( - \frac{1}{48} \pi(v) k \right) \\
                 &\leq \exp \left( - 4 \log n \right) ,
  \end{align*}
  where we have used the fact that $\pi(v) \geq n^{-2}$ and our lower bound on
  $x$ for the last inequality. Taking the union bound over all $n$ resources
  yields the claim for non-bipartite graphs.

  For bipartite graphs we argue similarly, but depending on whether $u$ or $v$
  are in the same partition or not, we either consider the round $t$ or the
  round $t+1$. By this we ensure that Equation \ref{eq:lower} still holds and
  we can use exactly the same arguments.
\end{proof}

  To prove the first term of the bound in Theorem~\ref{thm:centralMix}, we use the following lemma.
  \begin{lem}
    \label{lem:mixProgress}
    Given any starting state with profile $x$ and $\Phi(x) > 
    384\cdot\frac{|E|}{\delta}\cdot\log n$, let $R$ be the first round at 
    which a state with profile $x^R$ is reached such that
    \[
    \Phi(x^R) \; \le \; \left(1 - \frac{\varepsilon_{\min}
        \cdot\delta}{16d}\right)\cdot\Phi(x)\enspace.
    \]
	  It holds that $\Ex{R} = \Oh{\MIX(G)}$.
  \end{lem}
  \begin{proof}
  	For the starting profile $x$, we focus on the set 
	  $V_-(x) = \{v \mid x_v < \overline{T} - \Phi(x)/n\}$ of \emph{significantly underloaded} resources. 
	  In particular, a resource $v \in V_-(x)$ can stop an ``average'' share of the random walks 
	  in the system. The following lemma is proved in~\cite{AckermannFHS11}.
    \begin{lem}[\cite{AckermannFHS11}]
      If $\Phi(x) > 0$, then $|V_-(x)| \ge n\cdot\varepsilon_{\min}/2$.
    \end{lem}
    Now consider a starting state with $\Phi(x) \ge 384 \cdot
    \frac{|E|}{\delta} \cdot \log n$. Suppose first that the $\Phi(x)$ many 
    random walks are not stopped. Let us consider the resources in $V_-(x)$, 
    each of which has at least $\Phi(x)/n$ ``holes''. We first prove the lemma 
    for non-bipartite
    graphs. Note that for non-bipartite graphs, Lemma~\ref{lem:ballsBins}
    implies that after at most $2\MIX(G)$ steps in expectation we reach a
    single state $a^*$ with profile $x^*$ in which every resource carries at
    least $\frac{\delta}{4|E|} \cdot \Phi(x)$ random walks. Suppose we start
    stopping random walks exactly at state $a^*$, then because $\delta n \le
    2|E|$ we have that $\frac{\delta}{4|E|}\cdot\Phi(x)$ walks are removed on
    each resource in $V_-(x)$. In total, we know that
    \[
    \Phi(x^*) \; \le \; \left(1 - \frac{\varepsilon_{\min}\cdot\delta \cdot
        n}{8|E|}\right)\cdot\Phi(x)\enspace.
    \]

    Let us now take into account that random walks might stop before reaching
    their destination in state $a^*$. In particular, we let the system evolve
    exactly as before, however, we stop a random walk when its token is
    removed. Whenever a random walk is stopped early, this implies that the 
    potential drops by 1. We account only 1/2 towards the token. Consider all 
    random walk tokens that previously reached a resource $V_{-}(x)$ in the 
    unstopped process and got removed in the last iteration. If such a token 
    $t$ now reaches its  destination but is not removed, there is some other 
    token $t'$ that took the spot $t$ on its resource. In this case, we 
    account the other half of the potential decrease towards $t$. Otherwise, 
    $t$ was removed earlier and potentially took the spot of some other token. 
    Hence, in this case it also gets an accounted potential decrease of at 
    least 1/2. Thus, every such token receives an accounted potential decrease 
    of at least 1/2. We denote by $R$ the random variable that yields the time 
    step at which our process first arrives at a state with potential at most
 		\begin{eqnarray*}
      \Phi(x^R) & \le & \left(1- \frac{\varepsilon_{\min} \cdot \delta \cdot
          n}{16|E|}\right)\cdot\Phi(x) \; = \; \left(1- \frac{\varepsilon_{\min}
          \cdot \delta}{8d}\right)\cdot\Phi(x)\enspace.
    \end{eqnarray*}
    Obviously, we have $\Ex{R} = \Oh{ \MIX(G) }$.

    For bipartite graphs, we apply the same reasoning as above, however, we
    consider a combination of states $a^*$, $a^{**}$ described by
    Lemma~\ref{lem:ballsBins} that are reached after at most $2\MIX(G)+2$ and
    $2\MIX(G)+3$ steps in expectation. We stop random walks on one partition
    in $a^*$ and on the other partition in $a^{**}$ and consider the stopped
    random walks. By assuming that $\frac{\delta}{4|E|}\cdot\Phi(x)$ random
    walks are stopped on each resource in $V_{-}(x)$, we overestimate their
    real number at most by a factor of 2. Thereby, we lose an additional
    factor of 2 in comparison to the analysis for non-bipartite graphs above
    and obtain
    \[
    \Phi(x^R) \; \le \; \left(1- \frac{\varepsilon_{\min} \cdot
        \delta}{16d}\right)\cdot\Phi(x)\enspace.
    \]
  \end{proof}

\begin{proof}[Proof of Theorem~\ref{thm:centralMix}]
  By repeatedly applying the result of Lemma~\ref{lem:mixProgress}, it follows
  that we need in expectation a number of $\Oh{\frac{1}{\varepsilon_{\min}}
    \cdot \frac{d}{\delta} \cdot \MIX(G) \cdot \log(\Phi(x))}$ steps to reduce
  the potential to below $384 \cdot \frac{|E|}{\delta}\cdot \log n$. When we
  reach a state with potential below $384 \cdot \frac{|E|}{\delta}\cdot \log
  n$, we apply the ideas of Theorem~\ref{thm:centralHit} and get an additional
  convergence time of $\Oh{\hit(G) \cdot \ln(n)}$ in expectation. This proves
  the theorem.
\end{proof}
The following theorem shows that the bound in the previous theorem is essentially tight for our protocol. We will describe a class of graphs and starting states such that the convergence time of our protocol is characterized by the problem of moving a large number of users over a relatively sparse cut. This allows us to establish a lower bound using the mixing time. We note that our class of graphs encompasses instances that provide the lower bound of the theorem for every mixing time in $\Omega(n)$ and $O(n^2)$.
\begin{thm}
	\label{thm:lower}
	There is a class of graphs such that for above average thresholds the
	protocol converges to a balanced state in an expected number of
	$\Omega(\MIX(G) \cdot \ln(m))$ rounds.
\end{thm}
\begin{proof}
	Consider a graph $G$ that consists of two cliques $V_1, V_2$, each of size
	$n/2$. The two cliques are connected by a total of $k$ edges, where $1 \leq
	k \leq \frac{1}{5} n^2$. The edges between the cliques are distributed
	evenly, i.e., every vertex in each clique is connected to at least $\lfloor
	k/(n/2) \rfloor$ and to at most $\lceil k/(n/2) \rceil$ vertices in the
	other clique. There are $m \gg n$ users in the system, and we assume all
	thresholds of all resources and users are $T = \lceil (1+\varepsilon) \cdot
	\overline{T} \rceil$, for some small constant $\varepsilon > 0$. In the
	initial assignment $a$ all users are allocated to vertices in $V_1$ as
	follows. First we allocate to every resource in $V_1$ exactly $T$
	users. To one resource $v \in V_1$ with $\lfloor k/(n/2) \rfloor$ neighbors
	we then add all remaining $m - T \cdot (n/2)$ users. For sufficiently large
	$m$ and small $\varepsilon$, the initial load profile $x$ of this assignment
	yields $\Phi(x) \in \Omega(m)$. Thus, there are $\Theta(m)$ random walks in
	the graph, they all start at some vertex in $V_1$ with $\lfloor k/(n/2)
	\rfloor$ neighbors in $V_2$, and to reach a balanced state it is necessary
	that they all have to enter $V_2$ at least once.
	\begin{lem}
		\label{lem:cutMove}
		Consider a random walk that starts at a vertex in $V_1$ with $\lfloor
		k/(n/2) \rfloor$ neighbors in $V_2$. Then, for any integer $t$, the
		probability that the random walk stays within $V_1$ for $t$ steps is at
		least
		\[
 			4^{ - \frac{16 kt}{n^2} - \frac{1}{2} }.
		\]
	\end{lem}	
	
	\begin{proof}[Proof of Lemma~\ref{lem:cutMove}]
		We use an accounting argument to show that the random walk on $V$ does the
		same as a random walk restricted to $V_1$ for the first $t$ steps with the
		desired probability. 	
		
		Consider first a random walk restricted to $V_1$. For every visit to a
		vertex $u \in V_1$, the random walk obtains a credit of $|N(u) \cap V_2|$.
		Intuitively, the credit provides us with a measure on how much the random
		walk restricted to $V_1$ differs from the one on $V$ as it is closely
		related to the probability of leaving $V_1$ at vertex $u$. Let $C_t$ be
		the credit that a random walk obtains in step $t$. Our next claim is that
		$\Ex{\sum_{i=1}^t C_i} \leq t \cdot \frac{8k}{n}$. This is certainly true
		if $k/(n/2) \geq 1/4$, since $\Ex{\sum_{i=1}^n C_i} \le t\cdot \lceil 
		\frac{k}{n/2}\rceil \le 8tk/n$. Consider now the case where $k/(n/2) < 
		1/4$. In this case, at most $k$ vertices in $V_1$ are connected to $V_2$, 
		while the other $(n/2)-k$ vertices in $V_1$ are not connected to $V_1$.

		Hence, $C_0 = 1$ and for any $t \in \nat$
		\[
 			\Ex{C_t} \leq \left( \frac{n/2-k}{n/2-1} \right) \cdot 0 + \left(
 			\frac{k}{n/2-1} \right) \cdot 1 \leq \frac{2k}{n}.
		\]
		This establishes $\Ex{\sum_{i=1}^t C_i} \leq \frac{8kt}{n} $ and thus by
		Markov's inequality,
		\[
  		\Pro{\sum_{i=1}^t C_i \ge \frac{16kt}{n}} \le \Pro{ \sum_{i=1}^t C_i 
  		\geq 2 \Ex{\sum_{i=1}^t C_i} } \leq \frac{1}{2}.
		\]
		Consider now a random walk on $V$. Then the probability that the random
		walk on $V$ does the same as the random walk on $V_1$ for the first $t$
		steps is at least
		\[
 			\prod_{i=1}^t \left(1 - \frac{C_i}{n/2-1} \right) \geq
 			4^{-\sum_{i=1}^t \frac{C_i}{n/2}   } \enspace,
		\]
		where in the first inequality we have used that $C_i \leq n/4-2$.
		Hence, with probability at least $\frac{1}{2} \cdot 4^{-\frac{16kt}{n^2}}$
		the random walk on $V$ does the same as the random walk on $V_1$ and
		therefore does not leave $V_1$ during the first $t$ steps.
  \end{proof} 
	The probability that a single random walk has entered $V_2$ at least once
	after $t$ steps is at most $1-4^{ - \frac{16 kst}{n^2} - \frac{1}{2} }$.
	As walks are independent, with probability at least $1 - (1-4^{ -
	\frac{16 kt}{n^2} - \frac{1}{2} })^{\Phi(x)}$ at least one walk has
	remained in $V_1$ for the whole time, in which case we have not reached a
	balanced state. For
	\[
		t = \frac{\log_4 \Phi(x) - \frac{1}{2}}{16 k} \cdot n^2
	\]
	the latter probability is	$1 - (1 - \frac{1}{\Phi(x)})^{\Phi(x)} \ge
	1-\frac{1}{e}$. Therefore, the expected number of rounds needed to move all
	random walks to $V_2$ is in $\Omega(\ln(\Phi(x)) \cdot (n^2/k))$. Because 
	$\Phi(x) = \Theta(m)$ and $\MIX(G) = \Theta(n^2/k)$ the theorem follows.
\end{proof}

This shows that a factor $\ln(m)$ cannot be avoided if we want the protocol to reach a balanced state. However, the following theorem shows that, intuitively, the protocol balances most of the random walks on $G$ much faster. Suppose we first decrease all thresholds by a factor of $1/(1+n^{-\gamma})$. We balance with these adjusted thresholds for $\hit(G)\cdot\gamma \log n$ rounds and then continue with the larger original thresholds. This allows the unstopped random walks to balance quickly over the network. By increasing all thresholds we avoid that towards the end of the process many random walks have to reach a small subset of nodes. This allows to obtain a balanced state in a number of rounds that is even independent of $m$. The result holds for user-independent thresholds and for above average thresholds.

For simplicity, we slightly reformulate the approach. We assume to start with original thresholds and after $\hit(G)\cdot\gamma \log n$ rounds the protocol is allowed to increase all thresholds by a factor of $1/n^\gamma$. %
\begin{thm}
  \label{thm:fastApx}
  Consider user-independent thresholds or all above average
  thresholds. Let $\gamma \geq 1$ be any value. For an arbitrary
  starting state, after
  \[
  \Oh{ \hit(G) \cdot \gamma \log n}
  \]
  rounds we reach a state with profile $x'$, in which $\Phi_v(x') \le 10
  \Phi(x)/n^{\gamma}$ for every $v \in V$ with probability $1-(e/10)^{10
  \Phi(x) \cdot n^{-\gamma} }$. If we increase all thresholds by $\Phi(x)
  \cdot n^{-\gamma}$ at this round, we obtain a balanced state after a total
  of $\Oh{\hit(G) \cdot \gamma \log n}$ additional rounds in expectation.
\end{thm}

\begin{proof}
  Assign every token a random walk of length $\ell:=\hit(G) \cdot
  \gamma \log n$. Each random walk visits all nodes of $G$ with
  probability $1-n^{-\gamma}$. Let $\mathcal{A}$ be the event that at
  most $10 \Phi(x) \cdot n^{-\gamma}:=\rho$ random walks do not visit
  all nodes of $G$. Then,
  \begin{align*}
    \Pro{ \mathcal{A} } &\leq \binom{\Phi(x)}{10 \Phi(x) \cdot
      n^{-\gamma} } \cdot
    \left( n^{-\gamma} \right)^{10 \Phi(x) \cdot n^{-\gamma} } \\
    &\leq \left( e n^{\gamma} / 10 \right)^{10 \Phi(x) \cdot
      n^{-\gamma} } \cdot \left( n^{-\gamma} \right)^{10 \Phi(x) \cdot
      n^{-\gamma} } = \left( e/10 \right)^{ 10 \Phi(x) \cdot
      n^{-\gamma} }.
  \end{align*}

  For contradiction, suppose that at step $\ell$ there is a resource
  $v$ with potential $\Phi_v(x^\ell) \ge \rho + 1$. This implies that
  at least one of the tokens whose associated random walk visits all
  nodes of $G$ are placed on resource $v$ in $x^\ell$. On the other
  hand, this also implies that at step $\ell$ there is at least one
  resource $u \neq v$ with load less than $x_u^\ell \le \lceil m/n
  \rceil - 1$. If we now consider the random walk who is placed on $v$
  the latest (considering only random walks that visit all nodes in
  $G$), we obtain a contradiction, as the token of this random walk
  would have been removed when visiting some other resource, e.g., $u$
  instead of $v$.

  For the expected convergence time, we assume $\Phi(x) >
  n^{2\gamma}$, otherwise the result follows using
  Theorem~\ref{thm:centralHit}. Note that with a large probability we
  directly reach a balanced state. Otherwise, with probability
  $\left( e/10 \right)^{ 10 \Phi(x) \cdot n^{-\gamma}} <
  (1/3)^{\sqrt{\Phi(x)}}$ we need additional expected time of only
  $\Oh{\hit(G) \cdot \log(\Phi(x))}$.
\end{proof}

\section{User-Controlled Migration}
\label{sec:user}
In this section we consider a fully distributed protocol for the case of user-independent thresholds. In our protocol, in each round every user located on resource $v$ decides to migrate away from $v$ with a probability $p_v(x) = \alpha \cdot (\Phi_v(x)/T_v)$. If a user decides to migrate, it moves to a neighboring resource of $v$ chosen uniformly at random. We will assume $\alpha < 1$ in order to avoid trivial examples that may result in an infinite oscillation. This approach has the advantage that resources do not have to sort and control movements of users.
Yet, the resulting process closely resembles our resource-controlled protocol analyzed in the last section. We can again consider user migration in terms of random walks, but now the number of walks leaving a resource $v$ in a state is not exactly $\Phi_v(x)$ as in the resource-controlled protocol. In particular, we assume that each resource contains $\Phi_v(x)$ random walk tokens. Each user that decides to migrate, picks a token uniformly at random and takes it to its destination. One challenge of the user-controlled migration is that on certain networks there could be assignments such that expected potential value increases in the next round. This makes the analysis harder than the analysis of resource-controlled migration and also harder than the analysis of user-controlled migration on complete graphs, as in both cases the (expected) potential is always non-increasing.

We begin our analysis with some lemmas. First, the probability of a given token to move to another resource in a round can be bounded by $\Omega(\alpha)$. This is straightforward as tokens are indistinguishable and all users on a resource move with the same probability (see Lemma~\ref{lem:Schleife} in the Appendix). Hence, each random walk has a loop probability of $1-\Oh{\alpha}$.

When more than $\Phi_v(x)$ users migrate from $v$ in a round, this leads to creation of new random walks. We term each random walk created in this manner \emph{excess (random) walk} which leaves an \emph{artificial hole} on $v$. In contrast, we refer to \emph{ordinary random walks} and \emph{holes}. Observe that the creation of excess walks becomes quite unlikely, especially if $\Phi_v(x)$ is large.
\begin{lem}\label{lem:holes}
  Let $C \geq 1, t \in \mathbb{N}$ be any two values. Then with probability at
  least $1 - t \cdot n^{-(C \cdot (1-\alpha)/6)+1}$, no resource generates more
  than $C \log n$ excess random walks in each of the first $t$ rounds. Hence,
  in the first $t$ rounds, all resources generate in total at most $tn \cdot
  C \log n$ excess random walks with probability at least $1-t \cdot
  n^{-(C \cdot (1-\alpha)/6) + 1}$. Moreover, we generate in expectation at most
  $30\alpha^2 tn$ excess random walks in the first $t$ steps.
\end{lem}
%

\begin{proof}
  Consider a resource $v$ at any round $1 \leq s \leq t$ which is
  overloaded, i.e., its load is $x_v > T_v$. Then the number of agents
  $Z_v$ that leave $v$ has distribution $\Bin(x_v, \alpha\cdot
  \frac{x_v-T_v}{x_v} )$.  So, $\Ex{Z_v} = \alpha (x_v - T_v)$. Using
Lemma~\ref{lem:chernoff}, it follows that
  \begin{align*}
    \Pro{ Z_v \geq \Ex{Z_v} + \lambda } &\leq \exp \left(-
      \frac{\lambda^2 }{2 ( \Ex{Z_v} + \lambda/3) } \right).
  \end{align*}
  Choosing $\lambda = (1 - \alpha) x_v + C \log n$ for any $C \geq 1$
  yields
  \begin{align*}
    \Pro{ Z_v \geq x_v - T_v + C \log n} &= \Pro{ Z_v \geq
      \Ex{Z_v} + (1 - \alpha) (x_v  - T_v) + C \log n } \\
    &\leq \exp \left(- \frac{((1 - \alpha) (x_v - T_v) + C
        \log n)^2        }{2 ( \Ex{Z_v} +  (1 - \alpha) (x_v - T_v) + C \log n /3)        } \right) \\
    &\leq \exp \left(- \frac{((1 - \alpha) (x_v - T_v) + C \log n)^2 }{2 (
        \alpha (x_v - T_v) +
        (1 - \alpha) (x_v - T_v) + C \log n /3)        } \right) \\
    &\leq \exp \left(- \frac{((1 - \alpha) x_v + C \log n)^2 }{2 ( x_v
        + C \log n /3) } \right) \\
    &\leq \exp \left( - ( (1-\alpha) / 2 )^2 \cdot (2 x_v + 2(C/3) \log n) \right)
\\
    &\leq \exp \left( - C \log n \cdot (1-\alpha)/6 \right).
  \end{align*}
  Taking the union bound over all resources, $\Pro{ \exists v \colon\, Z_v 
  \geq x_v - T + C \log n} \leq n \cdot n^{-C \cdot (1-\alpha)/6}$. Finally, taking 
  the union bound over all time-steps up to time $t$, it follows that the 
  probability that up to time $t$ there is a time step $s$ in which one of the 
  resources $v$ has $Z_v \ge x_v - T_v + C \log n$ is at most 
  \[ \Pro{ \exists s \colon \, \exists \, v \colon Z_v \geq x_v - T + C \log 
  n} \leq t \cdot \cdot n \cdot n^{- C \cdot (1-\alpha)/6} = t \cdot n^{- (C \cdot (1-\alpha)/6) +1}\enspace.
  \] 
  This implies the first statement of the theorem. The result for the expected 
  value follows directly from~\cite[Lemma 2.4]{AckermannFHS11}.
\end{proof}

Our proofs rely on the condition that as long as the potential is above some value $\beta$, we have a multiplicative expected decrease. This way we obtain a state with potential of $\beta$ in expected time roughly logarithmic in the size of the initial potential. For a proof of the following lemma see the Appendix.
\begin{lem}
  \label{lem:stopping}
  Let $(X^t)_{t \in \nat}$ be a stochastic process with non-negative values 
  such that $\Ex{X^t} \leq (1-\gamma) \cdot X^{t-1}$ with $0 < \gamma < 1$ as 
  long as $X^{t-1} \geq \beta$.  Let $\tau:=\min \{t \in \mathbb{N}: X^t \leq 
  \beta \}$. Then $\Pro{ \tau \geq \frac{1}{\gamma} \cdot (1 + \ln ( 
  X^0/\beta)) } \leq 1/2 $.
\end{lem}
In the following two theorems we assume $\alpha = 1/(2e)$ and extend Theorems~\ref{thm:centralHit} and~\ref{thm:centralMix} to the scenario of user-controlled migration. The approach of the proofs is to bound the increase due to excess random walks and show that the potential (i.e., the number of random walks) still drops by a constant factor as long as the potential is sufficiently large.
\begin{thm}\label{thm:userhit}
  For feasible user-independent thresholds after $\Oh{\hit(G) \cdot \log m}$
  rounds in expectation we reach a state with profile $x$ where $\Phi(x) =
  \Oh{n \cdot \hit(G)}$.
\end{thm}
\begin{proof}
  We first consider only the ordinary random walks in the system and
  prove that after a fixed time interval, a significant fraction gets
  removed. Afterwards, we consider the effect of excess random walks
  and artificial holes during the interval. Let us consider a fixed
  time-interval of length $\ell:=2 \cdot \hit(G) \leq 2n^3$. As in the
  proof of Theorem~\ref{thm:centralHit} we assign tokens to holes and
  consider the unstopped random walk of a token within $\ell$
  rounds. The expected fraction of the tokens that reach their
  destination at least once during the interval is at least $1/2$ by
  Markov inequality. If $k$ random walks reach their destination, a
  similar argument as in Lemma~\ref{lem:hitProgress} shows that they
  contribute $k/2$ to the potential decrease. Finally, to account for
  excess walks and artificial holes, we note that during $\ell$ rounds
  there are, in expectation, at most $2\ell n$ excess random walks for
  our choice of $\alpha=1/(2e)$ (see Lemma~\ref{lem:holes}). Hence the
  expected value of the potential satisfies:
  \begin{align*}
    \Ex{ \Phi(x^{\ell}) \, \mid \, x^{0} } &\leq \frac{3}{4}
    \Phi(x^{0}) + 2\ell n.
  \end{align*}
  Conditioned on a load vector $x^{0}$ with $\Phi(x^{0}) \geq 16 \cdot n \cdot
  \hit(G)$, we have
  \begin{align*}
    \Ex{ \Phi(x^{\ell}) \, \mid \, x^{0} } &\leq \frac{3}{4}
    \Phi(x^{0}) + \frac{1}{8} \Phi(x^{0}) \leq \frac{7}{8} \Phi(x^0).
  \end{align*}

  Now we apply Lemma~\ref{lem:stopping} as follows. We consider a new
  iterative random process $Y^t$. $Y^t$ is the value of the potential after
  exactly $t\cdot \ell$ steps of the protocol. Then, the previous arguments
  show that $Y^t$ satisfies Lemma~\ref{lem:stopping} with $\gamma =
  \frac{1}{8}$ and $\beta = 16\cdot n \cdot H(G)$. Thus, after at most $\tau =
  8 \cdot (1+\ln(\Phi(x^{(0)})/(8n\hit(G))) \in O(\log(m))$ steps, the
  probability that the potential has dropped below $\beta$ is at least 1/2. By
  considering the process in blocks of length $\tau$, we see that the
  probability after $k$ blocks is at least $1-2^{-k}$, i.e., in expectation a
  constant number of blocks are needed. Hence, in expectation, the process
  $Y^t$ takes at most $O(\log(m))$ steps, so our protocol takes only
  $O(\hit(G) \cdot \log(m))$ steps. This proves the theorem.
\end{proof}
\begin{thm}\label{thm:usermix}
  For user-independent thresholds with $T_v \ge
  (1+\varepsilon_{\min}) \cdot \overline{T}$, after
  \[
  \Oh{\frac{1}{\varepsilon_{\min}}\cdot
    \frac{d}{\delta}\cdot \MIX(G) \cdot \log(m)}
  \]
  rounds in expectation we reach a state with profile $x$ where $\Phi(x) =
  \Oh{n\cdot\frac{d}{\delta \varepsilon_{\min}} \cdot \MIX(G)}$.
\end{thm}
\begin{proof}
  The proof is similar to Theorem~\ref{thm:userhit}, but this time we
  take the accounting approach of Theorem~\ref{thm:centralMix}. In
  particular, we consider a fixed time interval of $\ell=\MIX(G)$
  steps and let all random walks evolve without stopping. For
  non-bipartite graphs, Lemma~\ref{lem:ballsBins} shows that after
  $\ell$ steps, with probability at least $1-n^{-3}$, we have a
  significant load on each resource. Considering the significantly
  underloaded resources, this shows that if we stop random walks
  exactly in step $\ell$, we decrease the number of ordinary random
  walks by a factor of $\frac{\varepsilon_{\min}\cdot\delta \cdot
    n}{8|E|}$. A similar reaccounting argument as in
  Lemma~\ref{lem:mixProgress} shows that the real process, in which
  random walks are stopped earlier, achieves at least half of this
  decrease, i.e., a factor of $\frac{\varepsilon_{\min} \cdot
    \delta}{8d}$. Hence, as the number of ordinary random walks in the
  system only decreases, we have that, in expectation, after $\ell$
  rounds their number has decreased by a factor of at least
  $\frac{\varepsilon_{\min} \cdot \delta}{16d}$. For bipartite graphs,
  we consider two consecutive steps, which again leads to a slightly
  smaller decrease of at least $\frac{\varepsilon_{\min} \cdot
    \delta}{32d}$. Now to account for excess walks we again note that
  during $\ell$ rounds with $\alpha = 1/(2e)$ there are, in
  expectation, at most $2 \ell n$ excess random walks. Hence the
  expected value of the potential satisfies:
  \begin{align*}
    \Ex{ \Phi(x^{\ell}) \, \mid \, x^{0} } &\leq
    \left(1-\frac{\varepsilon_{\min} \cdot
        \delta}{32d}\right) \cdot \Phi(x^{0}) + 2\ell
    n.
  \end{align*}
  Conditioned on a load vector $x^{0}$ with
  \[
  \Phi(x^{0}) \geq \frac{128}{\varepsilon_{\min}} \cdot
  \frac{d}{\delta} \cdot \MIX(G) \cdot n
  \]
  we have
  \begin{align*}
    \Ex{ \Phi(x^{\ell}) \, \mid \, x^{0} } &\leq
    \left(1-\frac{\varepsilon_{\min} \cdot
        \delta}{32d}\right) \cdot \Phi(x^{0}) +
    \frac{\varepsilon_{\min} \cdot \delta}{64d} \cdot
    \Phi(x^{0}) \\
    & \leq \left(1-\frac{\varepsilon_{\min} \cdot
        \delta}{64d}\right) \cdot \Phi(x^{0})\enspace.
 \end{align*}
 Observe also that initial application of Lemma~\ref{lem:ballsBins}
 depends on
 \[
 \Phi(x^{0}) \geq 192 \cdot \frac{d}{\delta} \cdot n \log n\enspace,
 \]
 which is asymptotically a smaller bound, as $\MIX(G) = 4 \log
 n/\mu$.

 Exactly as in the proof of Theorem~\ref{thm:userhit}, we can now directly
 apply Lemma~\ref{lem:stopping} to show the theorem. In particular, we again
 define a process $Y^t$ that measures the potential after $t \cdot \ell$
 rounds and apply Lemma~\ref{lem:stopping} with suitable bounds to $Y^t$. By
 observing that in expectation only a constant number of applications of the
 lemma are needed, the statement in the theorem follows.
\end{proof}
The theorems do not guarantee convergence to a balanced state, because with small potential we are likely to create artificial holes and thereby increase the potential. It is, however, straightforward to derive with the proofs of the previous lemmas and theorems that convergence to a balanced state can be achieved by setting $\alpha = n^{-5}$. Then, by Lemma~\ref{lem:Schleife}, for every random walk the hitting time is increased by a factor of $1/\alpha$. Examining the proofs of Theorems~\ref{thm:centralHit} and~\ref{thm:userhit} shows that instead of $2\cdot \hit(G)$ rounds we get a removal of expected 1/4
of the ordinary random walks after $\Oh{\hit(G)/\alpha}$ rounds. On the other hand, observe that the expected number of excess random walks within $\ell$ rounds is only $30 \alpha^2 \cdot \ell \cdot n$. Thus, if we take $\alpha = n^{-5}$, this implies that the expected number of excess walks generated during $\Oh{n^5 \cdot \hit(G)}$ rounds are only $\Oh{1/n}$. Hence, we can adjust the lower bound on the potential in the final state in Theorem~\ref{thm:userhit} to 1 and spend an additional $\hit(G)/\alpha$ rounds for the last random walk to find a hole. Thus, we reach a balanced state with an additional factor of $n^5$ in the expected convergence time.
\begin{cor}
  For user-independent thresholds, suppose we set $\alpha:=n^{-5}$,
  then after an expected number of $\Oh{n^5 \cdot \hit(G) \cdot \log(m)}$
  rounds the protocol reaches a balanced state.
\end{cor}
A corresponding result similar to Theorem~\ref{thm:centralMix} holds
if we resort to Theorem~\ref{thm:usermix} and combine the result with
the above corollary.

If we do not want to slow down the protocol in this way and stick to $\alpha = 1/(2e)$, the following theorem shows that the process still rapidly balances all random walk tokens on the graph with high probability. The resulting state is not necessarily balanced. However, the overload in the allocation is balanced, i.e., in the resulting state every resource has a number of users exceeding its threshold by at most an average number.

\begin{thm}\label{thm:randomwalkverteilen}
  Suppose the process starts in a state $a^0$ with profile $x^0$. If
  $G$ is a regular, non-bipartite graph, then after $t= \MIX(G)$
  rounds, with probability $1-3 n^{-1}$ it holds for every resource $v$
  that
  \[
  \Phi_v(x^t) \leq \Phi(x^0)/n \cdot (1 + n^{-2}) + \lambda +
  \frac{294}{1-\alpha} \cdot \log^2 n \cdot (\MIX(G))^2,
  \]
  where $\lambda:= 8 \max\{ \sqrt{ 2(\Phi(x^0)/n \cdot (1 + n^{-2}) )
    \log n}, 4 \log n \}$. If $G$ is an arbitrary non-bipartite graph,
  then after $t= \MIX(G)$ rounds, with probability $1-2 n^{-1}$ it
  holds for every resource $v$ that
  \[
  \Phi_v(x^t) \leq \Phi(x^0) \cdot (\pi(v) + n^{-3}) + \lambda +
  \frac{294}{1-\alpha} \cdot n \log^2 n \cdot (\MIX(G))^2,
  \]
  where $\lambda$ is defined as above.
\end{thm}
The following lemma will be used for the proof of the theorem. In the lemma we focus on regular graphs, as on non-regular graphs a random walks does not converge to a uniform stationary distribution.
\begin{lem}\label{lem:randomwalkverteilen}
  Consider a time-interval $[1,\ell]$. Suppose that in each round $t
  \in [1,\ell]$, each node $v \in V$ on a regular graph $G$ generates at most
  $\rho \geq 1$ excess random walks. Then, with probability at least
  $1-n^{-1}$, no node is visited by more than $7 (\log n) \rho \ell^2$
  excess random walks in the time-interval $[1,\ell]$.
\end{lem}
\begin{proof}
  Fix a node $u \in V$ and let $Z_u$ denote the number of visits to
  $u$. Note that for a regular graph, the matrix $\mathbf{P}$ is
  symmetric and therefore each column sum equals one. This allows us
  to estimate
  \begin{align*}
    \Ex{Z_u} &= \sum_{t=1}^{\ell} \rho \sum_{v \in V} \sum_{s=1}^{\ell-t}
    P_{v,u}^{s} = \rho \sum_{t=1}^{\ell}\sum_{s=1}^{\ell-t} \left( \sum_{v \in
        V} P_{v,u}^{s} \right) \le \rho \ell^2.
  \end{align*}
  We now use the following Chernoff bound:
  \begin{align*}
    \Pro{ Z_u > (1+ \varepsilon) \Ex{Z_u} } &\leq e^{-\min\{ \varepsilon,
      \varepsilon^2 \} \cdot \mu/3}
  \end{align*}
  which yields for $\varepsilon = 6 \cdot \log n$ that $\Pro{ Z_u > (1+6\cdot
  \log n) \rho \ell^2 } \leq n^{-2}$. Taking the union bound over all nodes $u
  \in V$ finishes the proof.
\end{proof}
\begin{proof}[Proof of Theorem~\ref{thm:randomwalkverteilen}]
  Using Lemma~\ref{lem:holes} and $t = \MIX(G) \le 4 n^4 \log n < n^5$, we
  see that with probability at least $1-n^{-1}$, no node generates more than
  $\frac{42}{1-\alpha}\cdot\log n$ excess random walks during these $t$ 
  rounds. Applying Lemma~\ref{lem:randomwalkverteilen} with $\rho = 
  \frac{42}{1-\alpha}\cdot\log n$ and $\ell=\MIX(G)$ it follows that, with 
  probability at least $1-2n^{-1}$, no node is visited by more than 
  $294/(1-\alpha) \cdot \log^2 n \cdot (\MIX(G))^2$ excess random walks.

  Consider now an ordinary random walk token that starts in round $1$. If the
  random walk does not stop, we may apply Lemma~\ref{lem:mixinglemma} to
  conclude that for every node $v$ and starting node $u$ of the random walk,
  \[
  P_{u,v}^t \leq \pi(v) + n^{-3}.
  \]
  Let $X_v$ be the number of ordinary random walk tokens that are on node $v
  \in V$ at step $t$. Note that $X_v$ is stochastically smaller than $Y_v \sim
  \Bin(\Phi(x^0),\pi(v) + n^{-3})$. Clearly, $\Ex{Y_v} = \Phi(x^0) \cdot
  (\pi(v) + n^{-3})$. Using the Chernoff bound
  \begin{align*}
    \Pro{ |Y_v - \Phi(x^0) \cdot (\pi(v) + n^{-3})| > \lambda } &\leq
    \exp \left(- \frac{\lambda^2}{2(\Phi(x^0) \cdot (\pi(v) + n^{-3})
        + \lambda/3) } \right)
  \end{align*}
  and choosing $\lambda:= 8 \max\{ \sqrt{ 2(\Phi(x^0) \cdot (\pi(v) +
    n^{-3}) ) \log n}, 4 \log n \}$ gives
  \begin{align*}
    \Pro{ X_v \geq \Phi(x^0) \cdot (\pi(v) + n^{-3}) + \lambda } &\leq
    n^{-2}\enspace.
  \end{align*}
  Taking the union bound shows that with probability at least $1-n^{-1}$, $X_u
  \leq \Phi(x^0) \cdot (\pi(u) + n^{-3}) + \lambda$ holds for any node $u \in
  V$.

  Note that the load at node $v$ at round $t$ can be upper bounded by $X_v$
  (if an ordinary random walk stops earlier, it removes a hole) plus the
  number of excess random walks that visit $v$ during the first $t$ rounds.
  Taking the union bound, we find that the load at node $v$ at round $t$
  satisfies:
  \begin{align*}
    \Pro{ x_v^t \leq \lceil m/n \rceil + \Phi(x^0) \cdot (\pi(v) +
      n^{-3}) + \lambda + \frac{294}{1-\alpha} \, \log^2 n \, (\MIX(G))^2 }
      &\leq 1-3n^{-1}.
  \end{align*}

  The proof for non-regular graphs is the same, except that we estimate
  the number of excess random walks on a node $v \in V$ by the total
  number of extra random walks generated during the first $t$
  rounds.
\end{proof}

\section{Conclusion}

In this paper, we studied a new load balancing protocol in a
decentralized environment, where unsatisfied users decide
independently to jump ``blindly'' to a random neighboring resource. We
prove that this simple protocol achieves a convergence time which is
logarithmic in $m$ and polynomial in the hitting time (or mixing time)
of the underlying network.

The main open problem is to find improved upper bounds on the convergence times for certain graph topologies. While our lower bound in Theorem~\ref{thm:lower} holds for a variety of mixing times, establishing a matching lower bound for every graph structure remains an open problem. For certain networks (like the complete graph), there are protocols with user-controlled migration achieving convergence even in a time of roughly $\log \log m$~\cite{EvenDar05,Berenbrink07}. It would be extremely interesting if one can adjust our protocol to obtain similar results or extend these approaches to obtain doubly logarithmic bounds also for arbitrary networks.


\bibliographystyle{plain}
\bibliography{../../Bibfiles/literature,../../Bibfiles/martin}

\clearpage

\appendix
\section{Appendix}
\subsection{Technical Lemmas and Omitted Proofs}

\begin{lem}\label{lem:chernoff}
Let $X \sim \Bin(n,p)$ be a binomially distributed random variable. Then for any
$\lambda > 0$,
\begin{align*}
 \Pro{ |X - \Ex{X}| \geq \lambda } &\leq \exp \left(-        \frac{\lambda^2}{2 ( \Ex{X} +
\lambda/3) } \right)\enspace.
\end{align*}
\end{lem}

\vspace{0.5cm}

{{\noindent \bf Lemma~\ref{lem:mixinglemma}. \it Let $G$ be any graph, $u,v \in V$ be any two nodes and $t \geq 4
\log (n) / \mu $.
\begin{itemize}
 \item If $G$ is non-bipartite, then $P_{u,v}^t = \pi(v) \pm n^{-3}$.
 \item If $G$ is bipartite with partitions $V_1,V_2$, then
\begin{equation*}
P_{u,v}^t =
\begin{cases}
 \pi(v) \cdot (1+(-1)^{t+1}) \pm n^{-3} & \mbox{if $u \in
V_1, v \in V_2$ or $u \in V_2, v \in V_1$}, \\
 \pi(v) \cdot (1+(-1)^{t}) \pm n^{-3} & \mbox{if $u \in
V_1, v \in V_1$ or $u \in V_2, v
\in V_2$}.
\end{cases}
\end{equation*}
\end{itemize}
}

\begin{proof}[Proof of Lemma \ref{lem:mixinglemma}]
  For the result for non-bipartite graphs, see 
  e.g.~\cite[Chapter~12]{Levin09}. Let us now prove the result for
  bipartite graphs, where we follow the arguments from \cite{Lovasz93}
  for non-bipartite graphs.

  Denote by $\mathbf{A}$ is the adjacency matrix of $G$. Let $\mathbf{D}$ be 
  the diagonal matrix with diagonal entries $D_{u,u} = 1/\deg(u)$. Then the 
  matrix $\mathbf{N} := \mathbf{D}^{1/2} \mathbf{A} \mathbf{D}^{1/2} = 
  \mathbf{D}^{-1/2} \mathbf{P} \mathbf{D}^{1/2}$ is symmetric. Let $\mu_1 \geq 
  \mu_2 \geq \ldots \geq \mu_n$ be the eigenvalues of $N$ and 
  $g_1,g_2,\ldots,g_n$ be the corresponding eigenvectors of unit length. Then, 
  $g_{1,u} := \sqrt{\deg(u) / (2|E|) }$ defines an eigenvector of $\mathbf{N}$ 
  with eigenvalue $1$ (c.f.\ \cite{Lovasz93}). Similarly, it can be verified 
  that the vector $g_{n}$ defined by $g_{n,u} := \sqrt{\deg(u) / (2|E|) }$ if $u 
  \in V_1$ and $g_{n,v} := - \sqrt{\deg(v) / (2|E|) }$ if $v \in V_2$, is an 
  eigenvector with eigenvalue $-1$. The same argument also shows that 
  $\lambda_{n-1} > -1$, since $\lambda_2 < 1$ and $\lambda_{n-1} = 
  -\lambda_2$. Then using the spectral representation of $P_{u,v}^t$ we obtain
  \begin{align*}
  	P_{u,v}^t &= \sum_{k=1}^{n} \lambda_k^t g_{k,u} g_{k,v} \sqrt{	
  	\frac{\deg(v)}{\deg(u)}	}\enspace.
		\intertext{This can be rewritten using the definition of $g_{1}$ and
		$g_{n}$, and assuming $u \in V_1$ and $v \in V_2$ as:}
		P_{u,v}^t	&= \pi(v) + \sum_{k=2}^{n-1} \lambda_k^t g_{k,u}	g_{k,v}
		\sqrt{\frac{\deg(v)}{\deg(u)}	} + (-1)^{t} \cdot \sqrt{
		\frac{\deg(u)}{2|E|}} \cdot \left(- \sqrt{ \frac{\deg(v)}{2|E|} }
		\sqrt{\frac{\deg(v)}{\deg(u)}	}\right) \\
		&= \pi(v) + (-1)^{t+1} \pi(v) + \sum_{k=2}^{n-1}  \lambda_k^t g_{k,u}
		g_{k,v} \sqrt{\frac{\deg(v)}{\deg(u)}	}\enspace.
	\end{align*}
	We can bound the last summand by recalling that $\lambda_{n-1} > -1$ and 
	using Cauchy-Schwartz inequality,
	\begin{align*}
		\left| \sum_{k=2}^{n-1}  \lambda_k^t g_{k,u} g_{k,v} \sqrt{
		\frac{\deg(v)}{\deg(u)}	} \right| &\leq  (1-\mu)^t \cdot \sqrt{
		\frac{\deg(v)}{\deg(u)} } \cdot \sqrt{ \sum_{k=2}^{n-1} g_{k,u}^2 \cdot
		\sum_{k=2}^{n-1} g_{k,v}^2 } \\
		&\leq e^{-\mu t} \cdot \sqrt{\frac{\deg(v)}{\deg(u)} } \le n^{-4} \cdot 
		\sqrt{n^2} = n^{-3}
	\end{align*}
	as the eigenvectors $g_k$ were chosen to be of unit-length and using the 
	lower bound on $t$. The other cases, e.g., $u \in V_2$ and
	$v \in V_1$ are shown similarly.
\end{proof}

\begin{lem}
  \label{lem:Schleife}
  In one round of the protocol for user-controlled migration starting in a state with profile $x$,
  any given random walk token on resource $v$ with $x_v > T_v$ is
  moved with probability at least $\Omega(\alpha)$.
\end{lem}

\begin{proof}
  Note that the number of users that move from resource $v$ is given
  by a sum over independent Bernoulli variables with expectation
  $\alpha \cdot \Phi_v(x)$. Hence, Using Lemma~\ref{lem:chernoff} it
  follows that the probability that at most $\frac{\alpha}{2} \cdot
  \Phi_v(x)$ users move in one iteration is at most
  $\exp\left(-\frac{\alpha^2\Phi_v(x)^2}{8(\alpha \Phi_v(x) + (\alpha
      / 6) \Phi_v(x) )}\right) \le \exp\left( - \alpha \cdot \Phi_v(x)
    / 10 \right)$. Clearly, this probability can be upper bounded by
  some constant $<1$. In turn, this implies that with constant
  probability at least $\alpha\cdot\Phi_v(x)/2$ users migrate in one
  round, in which case the probability that a particular random walk
  token is moved is at least $\alpha/2$. Thus, the probability that a
  particular token is moved is $\Omega(\alpha)$.
\end{proof}

\vspace{0.2cm}

{{\noindent \bf Lemma~\ref{lem:stopping}.} \rm
  Let $(X^t)_{t \in \nat}$ be a stochastic process with non-negative values 
  such that $\Ex{X^t} \leq (1-\gamma) \cdot X^{t-1}$ with $0 < \gamma < 1$ as 
  long as $X^{t-1} \geq \beta$.  Let $\tau:=\min \{t \in \mathbb{N}: X^t \leq 
  \beta \}$. Then $\Pro{ \tau \geq \frac{1}{\gamma} \cdot (1 + \ln ( 
  X^0/\beta)) } \leq 1/2 $.
}

\begin{proof}[Proof of Lemma~\ref{lem:stopping}]
Let us define auxiliary random variables $Y^{t}$ by $Y^{0} := X^{0}$, and for any round $t \geq 1$,
\begin{align*}
  Y^{t} &=
  \begin{cases}
    X^{t} & \mbox{if $[X^{t-1} \geq \beta] \wedge [Y^{t-1} > 0]$} \\
    0 & \mbox{otherwise.}
  \end{cases}
\end{align*}
Then, for any $t \geq 1$, it holds $\Ex{ Y^{t} } \leq (1 - \gamma) \cdot Y^{t-1}$. We have for $\sigma = \frac{1}{\gamma} \cdot (1 + \ln(X^0/\beta))$ an expected value bounded by $ \Ex{ Y^{t} } \leq (1 - \gamma)^{\sigma} \cdot Y^{0} < \beta/e$. Hence by Markov's inequality $\Pro{ Y^{\sigma} \geq \beta } \leq 1/2$. We consider two cases.
\begin{description}
\item[\bf Case 1:] For all time-steps $t \in [1,\ldots,\sigma]$, $Y^t = X^t$. Then, by assumption $\Pro{ X^{\sigma} \geq \beta } \leq 1/2$.
\item[\bf Case 2:] There exists a step  $t \in [1,\ldots,\sigma]$ such that $Y^{t} \neq X^{t}$. Let $t$ be the smallest time step with that property. Since $Y^{0} = X^{0}$ by definition, $t \ge 1$. Hence, $Y^{t} \neq X^{t}$, but $Y^{t-1} = X^{t-1}$. If $Y^{t-1}=0$, then $X^{t-1}=0$. If $Y^{t-1} \neq 0$, then by definition of $Y^{t}$,
\[
 \left( Y^{t} \neq X^{t} \right) \bigwedge \left( Y^{t-1} \neq 0 \right) \Rightarrow X^{t-1} < \beta.
\]
\end{description}
In all cases we have shown that with probability at least 1/2 there exists a step $t \in [0,\sigma]$ so that $X^{t} < \beta$. This completes the proof of the lemma.
\end{proof}

\end{document}